\documentclass[envcountsame,runningheads,11pt]{llncs}
\usepackage[margin=1in]{geometry}
\usepackage{setspace}
\usepackage{amssymb,amsmath}
\usepackage{graphicx}
\usepackage{url}
\usepackage{paralist}
\usepackage{algorithm}
\usepackage{environ}
\usepackage{cite}
\usepackage{tikz}
\usetikzlibrary{shapes,calc,intersections}

\urldef{\mailsc}\path | {ashishg,anilesh}@stanford.edu |

\newtheorem{assumption}{Assumption}
\usepackage{bbm}
\usepackage{color}

\newcommand{\bfu}{\mathbf{u}}

\newcommand{\bfp}{\mathbf{p}}
\newcommand{\bfq}{\mathbf{q}}
\newcommand{\bfI}{\mathbf{I}}
\newcommand{\bfll}{\mathbf{L}}
\newcommand{\bfv}{\mathbf{v}}
\newcommand{\bfs}{\mathbf{s}}
\newcommand{\relDsim}{~\mathit{\mathcal{D}}\text{-}\mathrm{dominates}~}
\newcommand{\relD}{~\mathrm{strictly}~\mathcal{D}\text{-}\mathrm{dominates}~}
\begin{document}
% Title portion. Note the short title for running heads 
%\title[Implementing the lexicographic maxmin bargaining solution]{Implementing the lexicographic maxmin bargaining solution}  
%\author{Submission 213}

\mainmatter  % start of an individual contribution
\singlespacing
% first the title is needed
\title{Implementing the Lexicographic Maxmin Bargaining Solution}
%\title{Relating Distortion and Fairness of Social Choice Rules under Metric Preferences}

% a short form should be given in case it is too long for the running head
%\titlerunning{Relating Distortion and Fairness}

\author{Ashish Goel and Anilesh K. Krishnaswamy\\
\mailsc}
\authorrunning{Goel and Krishnaswamy}
%% (feature abused for this document to repeat the title also on left hand pages)

% the affiliations are given next; don't give your e-mail address
% unless you accept that it will be published
\institute{
Stanford University \\
}
\maketitle

% note that the abstract must come before \maketitle
\begin{abstract}
There has been much work on exhibiting mechanisms that implement various bargaining solutions, in particular the Kalai-Smorodinsky solution \cite{moulin1984implementing} and the Nash Bargaining solution. Another well-known and axiomatically well-studied solution is the lexicographic maxmin solution. However, there is no mechanism known for its implementation. To fill this gap, we construct a mechanism that implements the lexicographic maxmin solution as the unique subgame perfect equilibrium outcome in the n-player setting. As is standard in the literature on implementation of bargaining solutions, we use the assumption that any player can grab the entire surplus. Our mechanism consists of a binary game tree, with each node corresponding to a subgame where the players are allowed to choose between two outcomes. We characterize novel combinatorial properties of the lexicographic maxmin solution which are crucial to the design of our mechanism.
\end{abstract}

% note: this command has been disabled to remove the ACM copyright block. Sorry...
%\thanks{This work is supported by the National Science Foundation,
%  under grant CNS-0435060, grant CCR-0325197 and grant EN-CS-0329609.}

%%%%%%%%%%%%%
\section{Introduction}
\sloppy Bargaining situations are encountered everywhere. Wage negotiations between employers and trade unions, trade pacts between nations, and environmental negotiations between developed and developing countries are some examples of bargaining that we have seen over the years. Essentially, the problem is one of choosing a feasible alternative, through cooperation, by a group of entities with potentially conflicting preferences. As Kalai \cite{kalai1985solutions} points out, bargaining theory 
\begin{quote}
``may be viewed as a theory of consensus, because when it is applied it is often assumed that a final choice can be made if and only if every member of the group supports this choice."
\end{quote}

As bargaining theory inherently deals with the aggregation of peoples' preferences over a set of outcomes, it closely resembles the theory of social choice. One aspect that distinguishes the bargaining problem from social choice is the existence of a disagreement outcome which may be invoked (often by an outside mediator) when the parties involved in bargaining fail to reach consensus. We provide a formal definition of the bargaining problem in Section \ref{sec:model}.

The first question that bargaining attempts to answer is what constitutes a good way of choosing the consensus outcome. Assuming the preferences of the agents are known to the designer, what choice should be made? The earliest answer to this question was an axiomatic one: the well known Nash Bargaining solution \cite{nash1950bargaining}. Although originally defined for the 2-player case, the axiomatic characterization of the Nash solution extends seamlessly to the $n$-player case. 

Another prominent bargaining solution is the Kalai-Smorodinsky solution \cite{kalai1975other}, which is obtained by replacing Nash's independence condition by an axiom of monotonicity. The axiomatic characterization of this solution holds up only when there are two players, and fails to generalize to the case with $n$ players. A straightforward extension of the 2-player Kalai-Smorodinsky solution to the $n$-player case, one which follows the form of the 2-player solution, is not guaranteed to be Pareto optimal. Henceforth, we refer to this extension as the (n-player)  Kalai-Smorodinsky solution. 

In this paper, we are primarily interested in the lexicographic maxmin solution (see Section \ref{subsec:lexmaxmin} for a formal definition), as proposed by Sen \cite{sen2017collective} or Rawls \cite{rawls2009theory}. It is obtained by a repeated application of the maxmin criterion: first, selecting feasible outcomes that maximize the utility of the worst-off player, then, among these outcomes, selecting those that maximize the utility of the next worst-off player, and so on. The utility gains are measured with respect to the disagreement point. The lexicographic maxmin solution can also be characterized by a principled extension of the Kalai-Smorodinsky axioms to the n-player case \cite{imai1983individual}. Other characterizations of the lexicographic maxmin solution are also known \cite{chang1999characterization,chun1989lexicographic}. 

The lexicographic maxmin solution has also had a long history outside of the literature on bargaining. It corresponds directly to the notion of maxmin fairness \cite{bertsekas1992data,jaffe1981bottleneck} which has been extensively studied in network routing and bandwidth allocation problems \cite{kleinberg1999fairness,kumar2000fairness,goel2000combining}. It has also been studied in relation to the resource allocation \cite{ghodsi2011dominant,bogomolnaia2004random,parkes2015beyond} and cake-cutting \cite{chen2013truth} problems . Allowing for randomization, these settings are special cases of the bargaining problem we study in this paper, as discussed in Section \ref{sec:model}.

Besides defining good solutions, the other question which bargaining theory deals with is that of implementation -- designing a mechanism which, with no knowledge of the agents' preferences, still leads them to the desired solution. However, the players are assumed to have complete information of the preferences of other players. In other words, we need to construct a game for non-cooperative agents which results in the right solution in equilibrium. While many solution concepts could arguably be used in this scenario, the standard approach is to look at subgame perfect equilibria of repeated games. Bargaining inherently involves repeated interactions between players, of offers and counter offers, which makes the subgame perfect equilibrium a logical and natural choice for the solution concept to be employed \cite{herrero1992implementation}. Henceforth, we use the term \emph{implementation} to refer to the problem of designing a mechanism for which the desired bargaining solution is the unique subgame perfect outcome, with each player having complete information.

The implementation of most of the aforementioned bargaining solutions has been studied. The 2-player Nash solution can be approximately implemented \footnote{with time discounted utilities, in the limit as the discount factor goes to 0.} as the unique subgame perfect equilibrium in a game of alternating offers \cite{binmore1986nash}.  Of particular interest to us is the approach of Moulin \cite{moulin1984implementing} in implementing the Kalai-Smorodinsky solution. The same approach has been generalized to other solutions such as Nash (and its n-player extension), and relative utilitarian solutions \cite{miyagawa2002subgame,samejima2005note}. In fact, this approach works for any bargaining solution that is the maximizer of some convex function of the players' utilities. For instance, the Nash solution is equivalent to maximizing the product of utilities, the Kalai-Smorodinsky solution is equivalent to maximizing the minimum among the players' utilities, etc. In Section \ref{subsec:comparison}, we discuss the techniques from the above approach.

To the best of our knowledge, there is no known (subgame perfect) implementation of the lexicographic maxmin solution. There is some work on a repeated arbitration procedure \cite{bossert1995arbitration} which leads to the lexicographic maxmin solution in equilibrium. This procedure falls short of an \emph{implementation} in that it relies on the designer's having knowledge of the players' utility functions -- As pointed out by Bossert and Tan \cite{bossert1995arbitration},
\begin{quote}
``The assumption of utility information being available to the decision maker clearly is quite strong. \ldots For instance, it would be useful to consider an environment in which
the players know each other's utilities, but the arbitrator does not know the players' utility functions. "
\end{quote}
 We tackle exactly this problem, and by utilizing novel combinatorial properties of the lexicographic maxmin solution, we provide a mechanism for its implementation. We provide a quick overview of our results and technical contributions in Section \ref{sec:results}.

While our mechanism is technically intricate, there are three aspects of it that are interesting from a practical perspective.
\begin{enumerate}
\item Our subroutine mechanism (the Knockout mechanism, Section \ref{subsec:knockout}) is simpler, and interesting in its own right. When the players are given a choice between the lexicographic maxmin solution and any other outcome, the Knockout mechanism implements the former.
\item Our mechanism offers insight into the nature of offers and counter offers required to arrive at the desired solution.
\item The gameplay terminates in $\log n$ steps ($n$ is the number of agents) in equilibrium, i.e., when well-meaning and well-informed agents propose the right solution.
\end{enumerate}

Our mechanism used the standard assumption (Assumption \ref{ass:favorite}) that the space of outcomes is such that in any player's best outcome, all the surplus goes to her, and every one else gets no utility. This assumption is commonplace in the literature on implementation of bargaining solutions. For example, it is employed in the work by Moulin  \cite{moulin1984implementing}, and the various extensions of his approach \cite{miyagawa2002subgame,samejima2005note}. The bargaining problem without the above assumption is also interesting. Not much work has been done in this setting, besides the implementation of the Nash Bargaining solution \cite{howard1992social}. We leave as open the  problem of implementing any of the various bargaining solutions, including the lexicographic maxmin solution, absent the assumption that a single player can grab the entire surplus, or proving that such implementations are impossible.

\section{A summary of results and techniques}\label{sec:results}
The problem of interest in this paper can be stated as follows: Given $n$ players, and a convex space of outcomes, can we design a mechanism, i.e., a repeated game with complete information, for which the lexicographic maxmin solution (in reference to a disagreement point) is the unique subgame perfect outcome? Our main result is an affirmative answer to this question. 

As mentioned earlier, solving for a subgame perfect equilibrium in this case is necessitated by the fact that bargaining inherently involves offers and counter offers, and thereby a repeated game. Because the designer has no knowledge of the players' preferences, the typical building blocks of implementation mechanisms are proposals and threats. Nash equilibrium is an inadequate solution concept for it opens up the possibility of non-credible threats by players, leading to many meaningless equilibria \cite{rubinstein1982perfect}. On the other hand, implementation in dominant strategies is difficult for the bargaining problem because we have neither the possibility of side payments (as in VCG mechanisms) nor the explicit goal of divvying up a finite set of resources as in cake-cutting. The cake-cutting problem is a special case of our setting, with the disagreement point corresponding to no one getting any part of the cake. We must note here that for the cake-cutting problem with piecewise uniform utilities (where the utility of each agent for a unit portion of the cake is either a global constant or zero), it can be shown that the truthful mechanism of Chen et al.  \cite{chen2013truth} implements the lexicographic maxmin solution, even though they do not claim it explicitly. However, for piece-wise linear utilities, there is no truthful mechanism that implements the lexicographic maxmin solution, to the best of our knowledge. We also note that our mechanism is not coalition-proof. Similar to the above result of Chen et al. \cite{chen2013truth}, the strategy-proofness of the lexicographic maxmin solution has been established in a variety of domains with dichotomous preferences \cite{ghodsi2011dominant,parkes2015beyond,bogomolnaia2004random}. There is also very interesting work \cite{kurokawa2015leximin} generalizing, and so unifying, all of the above results.

We design our mechanism in two stages. First, we will consider a simpler problem (see Section \ref{sec:simple}) - given $n$ players who are offered a choice between the lexicographic maxmin solution $\bfu^*$ (with the same disagreement point as before), and some other outcome $\bfu$, is there a mechanism that implements $\bfu^*$? 

The mechanism from the above problem will be used as a subroutine in constructing our full mechanism. We describe this in Section \ref{sec:fullmechanism}. In short, we run a binary tree of games, where each node corresponds to a mechanism where the players to choose between just two outcomes. 

This outer mechanism is inspired by the mechanism of Howard \cite{howard1992social} which implements the Nash Bargaining solution, although the idea of using a binary tree is novel. A complete description of our game form requires a total number of rounds of the order of $n^2 \log n$ overall. Modeling our outer mechanism analogous to that of the mechanism of Howard \cite{howard1992social} (that implements the Nash solution) would lead to a description with $n^3$ rounds. In fact, our idea of the outer binary tree can be used to simplify Howard's mechanism \cite{howard1992social} for the implementation of the Nash bargaining solution.

To design our inner subroutine, we need a property that
\begin{itemize}
\item distinguishes the lexicographic maxmin solution from any other solution, and
\item can be exploited in a such a way that a game form based on it leads to $\bfu^*$ being necessarily chosen as the only subgame perfect outcome.
\end{itemize} 

Just the property that the lexicographic maxmin solution dominates every other outcome with respect to lexicographical ordering is inadequate for the above purpose (discussed in Section \ref{subsec:comparison}). Therefore, we identify a novel relation between any two vectors called \emph{disagreement dominance}. For any two vectors $\bfu, \bfv$, the disagreement dominance relation between the two is obtained by checking the lexicographical dominance relation between the ``disagreement projection" of $\bfu$ onto $\bfv$ and vice versa; where the disagreement projection of $\bfu$ onto $\bfv$ measures the utilities in $\bfu$ just for those players that prefer $\bfv$ over $\bfu$. The disagreement dominance relation has the very interesting property that the lexicographic maxmin solution disagreement dominates every other outcome\footnote{To the curious reader: disagreement dominance does not always agree with lexicographic ordering (see the example at the end of Section \ref{sec:disdom}).}.

%%%%%%%%%%%%%
\section{The bargaining problem: Model and Definitions}\label{sec:model}
We now describe the model of the bargaining problem, closely following Moulin \cite{moulin1984implementing}). We are given a finite set of alternatives denoted by $A$\footnote{This can be easily generalized to the case of infinite sets.}. These alternatives abstractly correspond to all the possible discrete outcomes. To make the overall space of outcomes convex, all lotteries over these discrete outcomes are added as possible outcomes. Let $P(A)$ be the set of probability distributions over $A$ :

\begin{align*}
P(A) = \{l \in \mathbb{R}^A | \; l_a  \geq 0 \; \forall a \in A, \; \sum_{a \in A} l_a = 1 \}.
\end{align*}

In other words, $P(A)$ is the set of all lotteries over $A$. An element $l$ of $P(A)$ is attained by a lottery that randomly selects outcome $a$ with probability $l_a$.	

The set of agents is denoted by $\mathcal{N} = \{1,2,\ldots, n\}$. Each agent $i \in \mathcal{N}$ has Von Neumann-Morgenstern preferences over $P(A)$, which can be represented via a vector $u^{i} \in \mathbb{R}^A$, and her utility for any lottery $l$ is given by:
\begin{align*}
U_i(l) = \sum_{a \in A} u^{i}_a l_a.
\end{align*}

An agent $i$ prefers a lottery $l$ over $l^\prime$ if and only if $U_i(l) > U_i(l^\prime)$. The agents are also assumed to be risk-neutral.

We are also given an exogenous lottery $s^*$, called the status quo outcome or disagreement point. This outcome represents the final outcome in the case where cooperation between agents breaks down, and they fail to reach a consensus. The designer, or an outside mediator has power to enforce this outcome in such a situation. We now formally define what is meant by a bargaining solution.

\begin{definition}[Bargaining Solution \cite{gaertner2009primer}]\label{def:barg_sol}
Given the set of alternatives $A$, the set of agents $\mathcal{N}$, and a status quo point $s^*$, a \emph{bargaining solution} is any function that maps any given utility profile $\bfu=(u^i)_{i \in \mathcal{N}}$ to a lottery $S(\bfu) \in P(A)$, such that 
\begin{itemize}
\item $S(\bfu)$ is individually rational, i.e., 
\begin{align*}
\forall i \in \mathcal{N}: ~ U_s(s^*) \leq U_i(S(\bfu)),
\end{align*}
\item $S(\bfu)$ is invariant under linear transformations of $\bfu$, i.e., for any scalars $\alpha_i > 0$, and $\beta_i$
\begin{align*}
v^i = \alpha_i u^i + \beta_i \mathbbm{1}, ~\forall i \in \mathcal{N} ~\implies S(\bfv) = S(\bfu),
\end{align*}
where $\mathbbm{1} \in \mathbb{R}^A$ is a vector with all components equal to $1$.
\end{itemize}
\end{definition}

For any utility profile $\bfu$, denote by $\mathbf{I}(\mathbf{u})$ the set of all individually rational lotteries at $\mathbf{u}$:
\begin{align}\label{eqn:IR}
\mathbf{I}(\mathbf{u}) = \{l \in P(A) | \; U_i(s^*) \leq U_i(l) \; \forall i \in \mathcal{N} \}.
\end{align}
We assume that $\mathbf{I}(\mathbf{u})$ is non-empty, and for every $\bfu$, limit the set of feasible outcomes to $\mathbf{I}(\bfu)$. 

We only allow those utility profiles $\bfu$ which satisfy the following assumption:

\begin{assumption}\label{ass:favorite}
For any agent $i$, all her favorite outcomes $l^*_i \in \arg \max_{l \in \mathbf{I}(\mathbf{u})} U_i(l)$ are such that  
\begin{align*}
U_j(l_i^*) &= U_j(s^*), \; &\forall j \neq i.
\end{align*}
\end{assumption}

This assumption has been indispensably used in the literature on implementation of bargaining solutions \cite{samejima2005note,moulin1984implementing}. It usually holds for bargaining over private resources such as in fair division and other problems. We now look at some specific problem instances where the above assumption holds.
\paragraph{Some motivating examples:} 
\begin{enumerate}[(A)]
\item Cake-cutting \cite{robertson1998cake,chen2013truth,brams1996fair}: Consider a heterogeneous cake given by the interval $[0,1]$. Each player $i$ has a value density function $v_i: [0,1] \to [0,\infty)$, and for any piece (a finite union of subintervals) of the cake $X$, her utility is given by $V_i(X) = \int_X v_i(x)dx$. The disagreement point is the outcome where no player gets any part of the cake. Moreover, if the players are \emph{hungry} \cite{branzei2015dictatorship}, i.e., the value density functions $v_i(.)$'s are strictly positive, then Assumption \ref{ass:favorite} is satisfied -- since the favorite outcome for each player is one where she gets the entire cake, and every other player gets nothing.
\item Wireless Relay Networks \cite{cao2012power,park2007bargaining,zhang2009fair}: Consider a wireless network with many agents communicating with their destinations via a single wireless relay. The relay has a total available power $P$, and this has to be divided among the agents. Denote the power allocated to agent $i$ by $P_i$, then its quality-of-service -- given by the effective SNR(Signal to Noise Ratio) -- can be described as
\begin{align*}
V_i(P_i) = \frac{a_i P_i}{b_i P_i + 1},
\end{align*}
where $a_i,b_i$ are agent-specific constants. We omit a discussion of the roots of the above expression (we point the curious reader to Cao et al. \cite{cao2012power} for details).

The disagreement point is given by an allocation of zero power to each agent. The function $V_i(P_i)$ is strictly monotone in $P_i$, whereby the best outcome for each agent is where she gets allocated all the power (i.e. $P_i = P$) -- thereby satisfying Assumption \ref{ass:favorite}.

\end{enumerate} 

Although we described our model starting with a finite (material) outcome space (and then taking its convex hull via lotteries), our results can be easily extended to the case where we start with a infinite set of outcomes. In this case, the space of utility outcomes (see Definition \ref{def:utilspace}) could be any general convex and compact space (as a subset of $R^n$). Given a convex and compact space of utility outcomes, Assumption \ref{ass:favorite} is also implied indirectly by certain restrictive but natural assumptions -- namely strict convexity and free disposal of utilities \cite{miyagawa2002subgame}. 

\subsection{Lexicographic maxmin solution}\label{subsec:lexmaxmin}
To satisfy the axiom of scale-invariance that every bargaining solution must satisfy (see Definition \ref{def:barg_sol}), we must use relative utility gains in our definitions. Instead, for the sake for convenience, and without loss of generality, let us assume the standard normalization of VNM utilities:

\begin{align}\label{eqn:normalization}
\forall i \in \mathcal{N},  ~ U_i(s^*) = 0, & \mbox{ and }
\underset{l \in \bfI(\mathbf{u})}{\max} U_i(l) = 1.
\end{align}

We now define what lexicographical ordering means in general. Given any vector $\bfp \in \mathbb{R}^n$, denote by $p_{(1)} \geq p_{(2)} \geq \ldots \geq p_{(n)}$ its components arranged in some non-increasing order. We say $\bfp \succ_L \bfq$ if and only if, for some $j$, $p_{(j)} > q_{(j)}$, and for $\forall i > j$, we have $p_{(i)} \geq q_{(i)}$. And $\bfp \sim_L \bfq$ if and only if $\forall i$, we have $p_{(i)} = q_{(i)}$. We say $\bfp \succeq_L \bfq$ if and only if $\bfp \succ_L \bfq$ or $\bfp \sim_L \bfq$.

\begin{definition} The lexicographic maxmin solution is a lottery $l^*$ such that  $\bfu(l^*) \succeq_L \bfu(l)$ for all $l \in p(A)$.
\end{definition}

The existence and uniqueness of such a solution is a consequence of the fact that the space $D = \{x \in \mathbb{R}^n : \exists l \in P(A), \mbox{ such that } x = [U_1(l),U_2(l), \ldots, U_n(l)]\}$ is convex (since P(A) is convex). We state the following lemma.

\begin{lemma}\label{lem:unique}
There exists a unique $\bfu^* \in D$ such that $\bfu^* \succ_L \bfu$ for any $\bfu \neq \bfu^*$.
\end{lemma}
It is easy to see that the above solution is guaranteed to be Pareto optimal. The Kalai-Smorodinsky solution in contrast is not guaranteed to be Pareto optimal.

We now illustrate the lexicographic maxmin solution with an example.

\begin{example}\label{eg:roth}
Consider a three person bargaining problem whose disagreement point $s^*$ is such that the utilities of the agents for it is $(0,0,0)$. The set of alternatives is given by $s^*$ and five other points $a,b,c,d,e$ for which the utilities are $(1,0,0)$, $(0,1,0)$, $(0,0,1)$, $(0.6,0.6,0.7)$, $(0.6,0.6,0)$ respectively. Clearly the lexicographic maxmin solution is $(0.6,0.6,0.7)$. The Kalai-Smorodinsky solution (see Appendix \ref{app:KS} for a definition) here is one that gives utilities $(0.6,0.6,0.6)$ to the agents, which is Pareto dominated. 
\end{example}

The lexicographic maxmin solution is uniquely characterized by a set of axioms generalizing those of Kalai-Smorodinsky \cite{imai1983individual}. Just for the case of two players, however, the two solutions coincide. For more than two players, they are different as seen in Example \ref{eg:roth}.

From here on forward, we will just use the convex space $D$ of utility outcomes, as opposed to dealing with the space of lotteries. Below, we summarize the relevant definitions in terms of $D$.
\begin{definition}\label{def:utilspace}
The space of utility outcomes is given by 
\begin{align*}
D = \{x \in \mathbb{R}^n : \exists l \in P(A), \mbox{ such that } x = [U_1(l),U_2(l), \ldots, U_n(l)]\}.
\end{align*}
\end{definition}

\begin{definition}[Lexicographic maxmin solution]\label{def:lexmaxmin}
The lexicographic maxmin solution is the unique point $\bfu^* \in D$ such that $\bfu^* \succ_L \bfu$ for all $\bfu \neq \bfu^*$.
\end{definition}

The following restatement of Assumption \ref{ass:favorite} in terms of the quantity $D$ will be more convenient for us in the following sections.
\begin{assumption}\label{ass:favorite2}
For any agent $i$, there exists a unique ideal point $\bfs^i \in D$, and moreover it satisfies   
\begin{align*}
s^i_i &= 1,\\
s^i_j &= 0, ~\forall j \neq i.
\end{align*}
\end{assumption}

The problem of implementing the lexicographic solution can be stated as follows:
\begin{definition}[Implementation]
Given $n$ players with complete information, collectively the set $\mathcal{N}$, and the utility space $D$, a mechanism $M(D,\mathcal{N},\bfs^*)$ is said to implement $\bfu^*$ if and only if $\bfu^*$ is the unique subgame perfect outcome of $M(D,\mathcal{N},\bfs^*)$.
\end{definition}

In other words, a mechanism $M(D,\mathcal{N}, \bfs^*)$ is said to implement $\bfu^*$ when $M(D,\mathcal{N},\bfs^*)$ always admits a subgame perfect equilibrium, and $\bfu^*$ is the outcome of any such equilibrium.

\subsection{Comparison with previous work}\label{subsec:comparison}
The implementation of the Kalai-Smorodinsky solution is possible via a mechanism that has a single round of simultaneous bidding, followed by a series of accept/reject decisions by the players \cite{moulin1984implementing}. The framework of Moulin's mechanism \cite{moulin1984implementing} has been extended to many other bargaining solutions \cite{samejima2005note,miyagawa2002subgame}.
In particular, it is known that any bargaining solution that is equivalent to maximizing a convex function of the players' utilities can be implemented as a subgame perfect equilibrium by a mechanism that does a pass over the entire list of players twice \cite{miyagawa2002subgame}. Such a mechanism works roughly as follows: In the first pass, each agent gets to propose an outcome, and a corresponding ``bid" (a vector of promised utilities to all the players). These bids are arranged in descending order according to the desired convex function, and the first player is designated the winner. Each player then gets to accept the winner's proposed outcome, or reject it and implement a dictatorial threat \cite{miyagawa2002subgame}. These games are much simpler than the one we come up with in this paper: the full description of these game involves a total number of rounds of the order of $n$.

Because the lexicographic maxmin solution does not correspond to maximizing a convex function of utilities, the above approach fails to work. A natural question is whether using lexicographic ordering to order the bids in the above approach solves our problem. In the remainder of this section, we will describe why this is not the case. We are therefore forced to look for more complicated game forms. A similar hurdle was faced by Howard  \cite{howard1992social} in trying to design a mechanism that implements the Nash solution without the standard assumption that every player grabs all the surplus in her favorite outcome, while the others get nothing.

Without going into the details of Moulin's mechanism (that implements the Kalai-Smorodinsky solution) \cite{moulin1984implementing} and other related ones, we would like to illustrate, by way of a concrete example, why direct extensions of these approaches do not work for the lexicographic maxmin solution. This will shed light on the complexity of our problem.

A direct extension of Moulin's mechanism \cite{moulin1984implementing} to our problem would roughly be the following: Each player $i$ declares a bid $\bfp^i \in [0,1]^n$, and these are ordered in accordance with lexicographic ordering. A player whose bid, say $\bfp$, lexicographically dominates every other bid is declared the winner, and gets to propose an outcome $\bfll$ in $D$. Each other player $i$ can either 
\begin{enumerate}
\item accept $\bfll$, or,
\item reject to impose the (lottery) outcome $p_i \bfs^i + (1-p_i) s^*$, which gives her a utility of $p_i$, as the final outcome.
\end{enumerate} 

If all players accept $\bfll$ then it becomes the final outcome. Note that when some player rejects, the utility for the winner, and indeed all others is $0$ by Assumption \ref{ass:favorite2}. 

While the lexicographic maxmin is a subgame perfect equilibrium of this game, it's not the only equilibrium. The fact that a player can lie about her own utility level in her bid (since it is not checked by anyone's threat) leads to many meaningless equilibria. Let's look at an example to see why this is the case.

\begin{example}\label{eg:moulinfail}
Recall the outcome space in Example \ref{eg:roth}: given 3 players, let the space $D$ be given by all points in the convex hull of the points
\begin{align*}
(0,0,0), (1,0,0), (0,1,0), (0,0,1), (0.6,0.6,0.7), (0.6,0.6,0).
\end{align*}
The disagreement point, as usual, is $(0,0,0)$.

The lexicographic maxmin solution here is $(0.6, 0.6, 0.7)$. 
\end{example}

In the above example, say player 2 bids $(0.6, 1, 0.65)$. Player 1 has no reason to challenge this bid, as she is guaranteed a utility of $0.6$, and there is no point where she can get a higher utility while ensuring that player 2 gets at least $0.6$. If player 3 bids above (lexicographically speaking) player 2's bid, and then proposes at outcome where she gets a utility more than $0.55$, her proposed outcome is bound to get rejected by someone. We now explain why this happens.

Even if player 3 bids $\bfp$ such that $p_3 = 1$\footnote{There's no reason to bid less than 1 here, as doing so can only hurt the lexicographic ordering of her whole bid.}, then the other components in her bid have to together lexicographically dominate $(0.6, 0.65)$. There is no outcome that ensures these levels of utility for the other players. Therefore, this leads us to an equilibrium where player 2 wins the bid and can possibly propose the outcome $(0.6, 0.6, 0.65)$ which is accepted by others.

Note that the above analysis holds even when the bids are done in turn. If the bids are simultaneous as in the game of Moulin's mechanism \cite{moulin1984implementing}, simpler examples of bad equilibria are possible. 

As a result, current techniques are inadequate to design a mechanism that implements the lexicographic maxmin. 
%%%%%%%%%%%%%%
\section{Implementing the lexicographic maxmin solution}\label{sec:simple}

\sloppy Before going into the general problem of implementing the lexicographic maxmin solution, we will look at a simpler problem: Given $\bfu^*$ and another outcome $\bfu$ (with the same disagreement point as before), can we design a mechanism which implements, for any given number of players, the solution $\bfu^*$? We answer this question by designing the Knockout mechanism. This mechanism will be an important subroutine in implementing our full mechanism for the implementation of the lexicographic maxmin solution.

Let us first characterize a novel combinatorial property of the lexicographic maxmin that will be of immense use to us.

\subsection{Disagreement dominance}\label{sec:disdom}
As mentioned before, given two utility outcomes, say $\bfu$ and $\bfv$, we use the \emph{disagreement-dominance} relation as a way of comparing the two.

For any pair of utility outcomes, $\bfu$ and $\bfv$, we define another vector $\pi(\bfu,\bfv)$ of the same size, which we call the \emph{disagreement} projection of $\bfu$ onto $\bfv$.
\begin{align*}
\left( \pi(\bfu,\bfv) \right)_i = \begin{cases}
u_i, \mbox{ if } u_i < v_i, \\
1, \mbox{ otherwise. }
\end{cases}
\end{align*}

We define the disagreement-dominance (\emph{$\mathcal{D}$-dominance}) relation as follows:
\begin{definition}[$\mathcal{D}\text{-}\mathrm{dominance}$]
$\bfu \relDsim \bfv$ if and only if $\pi(\bfu,\bfv) \succeq_L \pi(\bfv,\bfu)$. Moreover, if $\pi(\bfu,\bfv) \succ_L \pi(\bfv,\bfu)$, then $\bfu \relD \bfv$.
\end{definition}

The above relation is essentially comparing, in a leximin fashion, the utilities in $\bfu$ of those agents that prefer $\bfv$ over $\bfu$, and the utilities in $\bfv$ of those agents that prefer $\bfu$ over $\bfv$.

The following lemma, which characterizes the relationship between the lexicographical ordering and the $\mathcal{D}$-dominance relation, is a major technical contribution of this paper.

\begin{lemma}\label{lem:maximin-disagree}
The lexicographical maxmin solution $\bfu^*$ strictly $\mathcal{D}$-dominates any other feasible utility outcome $\bfu$, where $\bfu \neq \bfu^*$.
\end{lemma}
\begin{proof}
Assume to the contrary that $\bfu \relDsim \bfu^*$, i.e.,  $\pi(\bfu, \bfu^*) \succeq_L \pi(\bfu^*, \bfu)$.
We will show shortly that this assumption leads to the absurdity that $\frac{1}{2}( \bfu + \bfu^* ) \succ_L \bfu^* $.

Divide the players $\mathcal{\mathcal{N}} = \{1,2, \ldots, n\}$ into three sets $X,Y,Z$ as follows: 
\begin{align*}
X = \{ i \in \mathcal{N} : u_i > u^*_i \}, \\
Y = \{ i \in \mathcal{N} : u^*_i > u_i \}, \\
Z = \{ i \in \mathcal{N} : u_i = u^*_i \}.
\end{align*}

Without loss of generality we can assume that both $X$ and $Y$ are non-empty. Of course, both can't be empty simultaneously, for then $\bfu = \bfu^*$. If only $Y$ is empty, then $\bfu \succ_L \bfu^*$, which is an absurdity. If only $X$ is empty, then $\bfu^* \relD \bfu$ trivially.

Let $u^*_{(1)} \leq u^*_{(2)} \leq \ldots \leq u^*_{(n)}$ denote the components of $\bfu^*$ arranged in non-decreasing order , and let $\sigma(k)$ be the agent corresponding to $u^*_{(k)}$.

Let $j$ be the smallest value such that $\sigma(j) \in X \cup Y$, and let $u^*_{\sigma(j)} = \alpha$. Since ties can be broken arbitrarily in the above ordering of the components, assume without loss of generality that
\begin{align}\label{eqn:accouting}
\forall j^\prime \geq j, ~u^*_{\sigma(j^\prime)} > \alpha \mbox{ whenever }j^\prime \in Z.
\end{align}

By definition, for $j^\prime < j$, we have $u^*_{\sigma(j)} = u_{\sigma(j)}$. In addition, if we show that for $j^\prime \geq j$, $\frac{1}{2}(u^*_{\sigma(j)} + u_{\sigma(j)}) > \alpha$, then $\frac{1}{2}(\bfu^* + \bfu) \succ_L \bfu^*$ which is a contradiction according to Lemma \ref{lem:unique}.

Let $\gamma = \min_{i \in X} u^*_i$, and $\delta =  \min_{i \in Y} u_i$. Since $\alpha = \min_{i \in X \cup Y} u^*_i$, we have $\gamma \geq \alpha$.
Since $\pi(\bfu, \bfu^*) \succeq_L \pi(\bfu^*, \bfu)$, we have $\min_{i \in Y} u_i \geq \min_{i \in X} u^*_i$, i.e., $\delta \geq \gamma$. Combining the above two arguments, we have $\delta \geq \alpha$.

In other words, we have
\begin{align*}
\forall i \in X, \; u^*_i \geq \alpha, \\
\forall i \in Y, \; u_i \geq \alpha.
\end{align*}

Now consider any $j^\prime \geq j$. There are three possible cases here: \\
\emph{(A) $\sigma(j^\prime) \in X$:} In this case, $u_{\sigma(j^\prime)} > u^*_{\sigma(j^\prime)}$, and $u^*_{\sigma(j^\prime)} \geq \alpha$. Therefore $\frac{1}{2}(u^*_{\sigma(j^\prime)} + u_{\sigma(j^\prime)}) > \alpha$.

\emph{(B) $\sigma(j^\prime) \in Y$:} In this case, $u^*_{\sigma(j^\prime)} > u_{\sigma(j^\prime)}$, and $u_{\sigma(j^\prime)} \geq \alpha$. Therefore $\frac{1}{2}(u^*_{\sigma(j^\prime)} + u_{\sigma(j^\prime)}) > \alpha$.

\emph{(C) $\sigma(j^\prime) \in Z$:} By our assumption to simplify accounting (equation \ref{eqn:accouting} above), we already have $u^*_{\sigma(j^\prime)} > \alpha$, which automatically means $\frac{1}{2}(u^*_{\sigma(j^\prime)} + u_{\sigma(j^\prime)}) > \alpha$, since $u^*_{\sigma(j^\prime)} = u_{\sigma(j^\prime)}$. \qed
\end{proof}

We note here that lexicographic dominance does not imply $\mathcal{D}$-dominance. For example, take $\bfu = (0.1,0.2,0.3,0.4,0.5)$ and $\bfv = (0.3,0.3,0.4,0.1,0.2)$. Here $\bfu \succ_L \bfv$ but $\bfv \relD \bfu$. %The property of $\mathcal{D}$-dominance is guaranteed only when we consider the lexicographic maxmin solution $\bfu^*$.

\subsection{The Knockout mechanism}\label{subsec:knockout}
 Armed with Lemma \ref{lem:maximin-disagree}, we first look for a mechanism where the players are given a choice between $\bfu^*$ and some other outcome $\bfu$, with the same disagreement point as before, i.e., $\bfs^*$. Before we describe this game, we will need to define the mechanism $D(i,\bfp)$, agent $i$'s dictatorial mechanism \cite{moulin1984implementing}, which comes into force when $i$ rejects a proposal in the overall mechanism. 
 
\begin{definition}[$D(i,\bfp)$]\label{def:dict}
Agent $i$ proposes an outcome $\bfll^i$ for unanimous approval by other agents. If at least one agent rejects $\bfll^i$, the referee enforces the status quo $\bfs^*$. If everybody accepts $\bfll^i$, the referee enforces the lottery $p_i \bfll^i + (1-p_i)\bfs^*$.
\end{definition}

 A straightforward consequence of Assumption \ref{ass:favorite2} is that the outcome $\bfs^i$, i.e., $i$'s ideal outcome, is a subgame perfect proposal of $i$ in $D(i,\bfp)$. One might wonder why we could not just have a provision to directly implement $p_i \bfll^i + (1-p_i)\bfs^*$ based on player $i$'s choice of outcome. The point is that $\bfll^i$ could be Pareto-dominated by $s^*$, and we need a device to ensure individual rationality, with respect to $s^*$ being normalized to $(0,0,\ldots,0)$, when some agent rejects a proposed lottery.
 
We now define the Knockout mechanism $K(X, Y, \mathcal{N})$ on a pair of outcomes $X,Y$, and a set of players $\mathcal{N}$ of size $n$. The game will be defined recursively. \\

\textbf{Knockout Mechanism $K(X,Y,\mathcal{N})$}\footnote{Recall that $X$ and $Y$, like all feasible solutions, satisfy individual rationality with respect to $\bfs^*$ (see equation \ref{eqn:IR}), i.e., no player strictly prefers $\bfs^*$ over either X or Y.}

\emph{Round 1:} 
\begin{enumerate}[(A)]
\item Each players $i$ in turn chooses her preferred outcome $\bfll^i$ to be either $X$ or $Y$.
\item Each player $i$ bids a vector $\bfp^i \in [0,1]^n$ in turn\footnote{Assuming $i$ prefers $X$ over $Y$, her bid $\bfp^i$ here can be thought of as the levels of utility she promises to players who prefer $Y$ over $X$. Also, as mentioned in the proof of Lemma \ref{lem:indbase}, simultaneous bidding here leads to the existence of bad equilibria.}.
\item The players are re-ordered, with all ties broken lexicographically, such that 
\begin{align*}
&\bfp^1 \succeq_L \bfp^i, ~\forall ~i\neq 1, \mbox{ and } p^1_2 \leq p^1_3 \leq \ldots p^1_n,
\end{align*}
where $p^i_j$ is the component corresponding to player $j$ in the bid $\bfp^i$ made by player $i$.
\end{enumerate}

\emph{Round $j$ ($j \in \{2,3,\ldots,n\}$):}
Player $j$ has to make a choice between two options which vary based on the value of her value in player $1$'s winning bid, i.e., $p^1_j$:
\begin{enumerate}[(A)]
\item If $p^1_j < 1$, then player $k$ can either accept $\bfll^1$, or reject it. In the latter case, i.e., a rejection by player $j$, the game immediately reduces to $D(j,\bfp^1)$. 
\item If $p^1_j = 1$, the player $k$ can either accept $\bfll^1$, or reject it. In the latter case, i.e., a rejection by player $j$, the game immediately reduces to $K(X,Y,\mathcal{N}-\{1\})$ (with player $1$ ejected).
\end{enumerate}

If player $j$ accepts, then we move on to the next player $j+1$, and so on, to player $n$ in round $n$. 
If all players accept and we reach the end of round $n$, then $\bfll^1$ is the final outcome.

\subsection{Analyzing the Knockout mechanism}
The above mechanism implements the lexicographic maxmin solution when it is pitted against any other outcome. More formally, when one of the outcomes is $\bfu^*$, the above game admits a subgame perfect equilibrium, and that any subgame perfect equilibrium of this game results in a payoff of $u^*$ to the agents. In fact, we will prove a much stronger property of the Knockout mechanism in terms of the $\mathcal{D}$-dominance relation.

\begin{theorem}\label{thm:indextend}
If $X \relD Y$, then $K(X,Y,\mathcal{N})$ implements $X$.
\end{theorem}

%\begin{theorem}\label{thm:knockout_LM}
%For any $\bfu \neq \bfu^*$, the mechanism $K(\bfu^*, \bfu, \mathcal{N})$ implements $\bfu^*$.
%\end{theorem}

We first prove the above theorem for the case of $n = 2$, which will serve as the base case for an inductive proof of the theorem in all generality. Given two players $\{1,2\}$, and two outcomes $X, Y$ \footnote{Even though we have only 2 players, we let $X,Y \in [0,1]^n$ to denote outcomes.}, let $X^{\{1,2\}} = (u_1,u_2)$ and $Y^{\{1,2\}} = (v_1,v_2)$.
\begin{lemma}\label{lem:indbase}
If $X^{\{1,2\}} \relD Y^{\{1,2\}}$, then $K(X,Y,\{1,2\})$ implements $X$.
\end{lemma}
\begin{proof}
Assume that ties are broken in favor of player 2.
In Round 1, let's say player 1 chooses $\bfll_1 = A$ and player 2 chooses $\bfll_2 = B$. The most pertinent case is when $1$ strictly prefers $A$, and $2$ strictly prefers $B$. All other cases follow analogously.

\emph{Case (1) $A \relD B$:} Now player 1 can bid $\bfp^1 = \pi(A^{\{1,2\}},B^{\{1,2\}}) = (1,A_2)$. To outbid player 1, player 2 has to bid $\bfp^2 = (t,1)$ where $t \geq A_2$. Since $A^{\{1,2\}} \relD B^{\{1,2\}}$, we know that $\pi(A^{\{1,2\}},B^{\{1,2\}}) \succ_L \pi(B^{\{1,2\}},A^{\{1,2\}}) \iff (1,A_2) \succ_L (B_1,1) \iff A_2 > B_1$. Therefore, player 2's bid must be such that $t > B_1$, in which case, player 1 is better off by rejecting, and player 2 gets 0 utility in this case. Consequently, $A$ is the outcome in this case.

\emph{Case (2) $B \relDsim A$:} Now player 1 has to specify some bid $\bfp^1 = (1,s)$ such that $s \leq A_2$, for otherwise, player 2 will reject and player 1 will get a utility of 0. Recall that ties are broken in favor of player 2. So player 2 can now outbid player 1, by specifying a bid $\bfp^2 = (t,1)$ such that $B_1 \geq t \geq A_2 \geq s$. This is because we know that in this case, $B \relD A \implies B_1 \geq A_2$. And in this case, $B$ is the outcome, since player 1 will accept.

Putting these together, we see that in any subgame perfect equilibrium of the game from Round 1 on, if $A \relD B$, then $A$ is the outcome (and if not, $B$ is the outcome).

As a result, since $X \relD Y$, we know that at least one of the players strictly prefers $X$ over $Y$. Consequently, $X$ is the only possible subgame perfect outcome. 

To show that a subgame perfect equilibrium exists, we show that a finite backwards reduction is possible. Let's say we are given the choice of $\bfll^1$ and $\bfll^2$ in Round 1(A). While infinite choices are available for the bids in the round 1(B), each player would not make a bid that gets rejected in favor of the dictatorial mechanism\footnote{This might not hold with simultaneous bidding in the case of more than two players. The equilibrium bids of the two players could be such that if one lowers his bid the other wins, and a third player rejects regardless of who wins.}. In effect, player 1, for example, is in effect making a choice between his own choice $\bfll_1$, player 2's choice $\bfll_2$, and the outcome of the subgame $L(X,Y,\{2\})$ (which by definition is equal to player 2's preferred outcome between $X$ and $Y$). As the number of choices of outcome in each step is effectively finite, the game can be reduced backwards, thereby guaranteeing the existence of a subgame perfect equilibrium. \qed
\end{proof}

In the case when $X \relDsim Y$ and $Y \relDsim X$, the outcome depends on how ties are broken between agents in the bidding phase of the game, and which outcome the winning bidder prefers.

We now extend the above claim via an inductive argument to prove Theorem \ref{thm:indextend}.

\renewcommand*{\proofname}{Proof of Theorem \ref{thm:indextend}}
\begin{proof}
This proof is essentially similar to that of Lemma \ref{lem:indbase}, with some additional complexities. In what follows, we highlight these differences, and discuss how we can take care of them. As mentioned before, the proof is an inductive one: an induction on the size of $\mathcal{N}$. Let the above statement be true for all sets of agents smaller than $\mathcal{N}$. The base case when $|\mathcal{N}| =2$ holds by Lemma \ref{lem:indbase}.

Now consider $K(X,Y,\mathcal{N})$. Let $A$ be the set of agents that strictly prefer $X$, and $B$ be the set of agents that strictly prefer $Y$, and $C$ be the set of agents that are indifferent. $X \relD Y$ necessitates that $A$ is non-empty. If $B$ is empty, then the argument below holds vacuously. So assume that $B$ is non-empty.

Consider an agent $i$ in $A$. This agent can choose $\bfll^i = X$ and bid $\bfp^i = \pi(X,Y)$. If this bid wins, then any player $k$ in $A$ cannot do better than accepting. Recall that, by way of the labeling in Round 1(C), we go through the players in increasing order of their components in the winner's bid vector. Therefore, only players in $A$ follow $k$.  Also, since $p^i_j = X_j$ for all $j \in B$, any agent in $B$ cannot do better than accepting either. This follows from the fact that $X_j = p^i_j < 1$ for all $j \in B$, and by rejecting, and thereby reducing the game to $D(\bfp^i,j)$ (recall Round $j$(A), for any $j \geq2$, in the definition of the Knockout mechanism), any such agent $j$ can get at most $X_j$.

Any agent $j$ in $B$ does not benefit by bidding any $\bfp$ such that $\bfp \succ_L \pi(Y,X)$, because in this case, there exists a player $k \in A$ such that $p_k > X_k$. Here we have to look at two cases when $j$ wins: 
\begin{enumerate}[(A)]
\item  If $p_k = 1$ then the player $k$ will reject $j$'s proposal and the game will be reduced to $K(X,Y,\mathcal{N}-\{j\})$. The outcome of this game by the induction assumption is $X$, since $X_j < Y_j$, and removing $j$ from the set of players still ensures $X^{\mathcal{N}-\{j\}} \relD Y^{\mathcal{N}-\{j\}}$ \footnote{$X^S$ denotes the projection of the $n$-dimensional $X$ onto the subset of coordinates corresponding to the set of players $S$.}.
\item If $p_k < 1$, then player $k$ will reject and the game reduces to $D(\bfp,k)$, the outcome of which has 0 utility for player $j$.  
\end{enumerate}
Thus, in either case, the outcome for player $j \in B$ is no better than $X$. Hence, $j$ would rather accept $i$'s proposal of $X$.

Combining the arguments above, we have that any subgame perfect equilibrium will lead to the outcome $X$. 

The existence of a subgame perfect equilibrium follows from the fact that a backwards reduction is possible, using an argument similar to that for Lemma \ref{lem:indbase}. \qed
\end{proof}
\renewcommand*{\proofname}{Proof}
A corollary of the above result, by way of Lemma \ref{lem:maximin-disagree}, is that when one of the two allowed outcomes is $\bfu^*$, the Knockout mechanism implements $\bfu^*$.

\begin{corollary}\label{thm:knockout_gen}
For any outcome $\bfu \neq \bfu^*$, the mechanism $K(\bfu^*,\bfu,\mathcal{N})$ implements $\bfu^*$.
\end{corollary}

\section{The outer binary tree mechanism: a recursive system of Knockout games}\label{sec:fullmechanism}
\sloppy
In this section, we develop a complete mechanism -- \emph{a binary game tree system with recursive Knockout games} -- that implements the lexicographic maxmin solution. On a sidenote, a modified version of our binary game tree can be used to simplify Howard's mechanism  \cite{howard1992social} that implements the Nash Bargaining solution.

Before describing our outer binary tree mechanism, we extend the definition of the Knockout mechanism to a game over two subgames, instead of two outcomes. We had defined $K(X,Y,\mathcal{N})$ as the Knockout mechanism over two outcomes $X$ and $Y$. Similarly given any two bargaining games $A$ and $B$ that each produce some outcome, we define the Knockout mechanism $K(A,B,\mathcal{N})$ analogously: by replacing the outcomes $X$ and $Y$ in the game tree of $K(X,Y,\mathcal{N})$ by the pending outcomes of subgames $A$ and $B$. In other words, instead of bargaining over the outcomes, the players bargain over which of the  subgames $A$ or $B$ they prefer to play next. We now make the following simple observation:

\begin{lemma}\label{lem:KMsubgames}
If $\alpha$ and $\beta$ are the unique subgame perfect outcomes of games $A$ and $B$, and $\alpha \relD \beta$, then $\alpha$ is the unique subgame perfect outcome of $K(A,B,\mathcal{N})$.
\end{lemma}
The above lemma is a simple consequence of the backward induction property (see, e.g. Proposition 99.2 in Osborne and Rubinstein \cite{osborne1994course}) of a subgame perfect equilibrium. We include a proof in Appendix \ref{app:subgames} for the sake of completeness. 

\subsection{Defining the outer binary tree mechanism}
The outer binary tree mechanism can be defined (see Figure \ref{fig:binary_tree}) as follows:	
\begin{itemize}
\item First, each player $i$ in turn chooses an outcome of her choice $P_i$.
\item For the sake of simplicity, assume that $n$ is a power of 2. The game tree\footnote{The game tree can be modified very easily if $n$ is not a power of 2. Details are deferred to a full version.}$^{,}$\footnote{Also, we are describing not the extensive form of the game tree, but just a scaffold -- with Knockout games placed on each of the nodes of the tree. On the whole, we have a system of recursive Knockout games.} is obtained as follows: create a source node $S_{1,2,\ldots, n}$, then its child nodes $S_{1,2,\ldots, n/2}$ and $S_{n/2 +1,\ldots, n}$, and so on, till we reach the leaf nodes $S_1, S_2, \ldots, S_n$. 
\item Each leaf node $S_i$ is identified with the outcome $P_i$.
\item Given any node $X$ in the above tree, and its child nodes $A$ and $B$, the game at $X$ is defined recursively by the Knockout mechanism $K(A,B,\mathcal{N})$. 
\item The game-play starts at $S_{1,2,\ldots,n}$, the game $K(S_{1,2,\ldots,\frac{n}{2}},S_{\frac{n}{2} + 1,\ldots, n},\mathcal{N})$ is played, and if\footnote{The game could also go into a player's dictatorial mechanism (recall Definition \ref{def:dict}) and end there.} and when one of $S_{1,2,\ldots,\frac{n}{2}}$ or $S_{\frac{n}{2} + 1,\ldots, n}$ is the chosen outcome, the game moves to the corresponding node, and so on.
\end{itemize}

Note that all players are involved at the start of the game-play at every node. The subscripts of the nodes indicate which of the $S_i$'s are reachable from each of them, and not the players involved. However, some agents might be temporarily not involved, till the game reaches the next node, or a player's dictatorial mechanism (recall Definition \ref{def:dict}).

Before analyzing the above mechanism, we make an important note on the intuition behind its working. 
\emph{Our mechanism, in its entirety, is a recursive system that works like a knockout tournament in reverse. By a repeated application of Lemma \ref{lem:KMsubgames}, doing a backward induction on the game tree above is equivalent to running a standard knockout tournament (\`a la a tennis grand slam) on the set $\{P_i:i \in \mathcal{N}\}$, where each pairwise game is decided according to the $\mathcal{D}$-dominance relation between the outcomes involved.}

\begin{figure}
    \centering
    \includegraphics[scale=0.7]{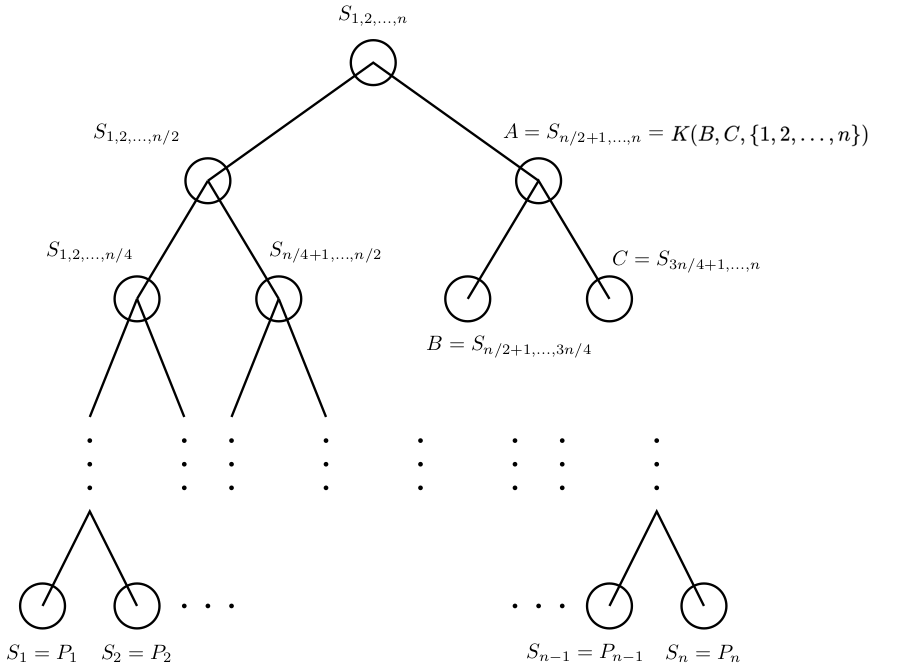}
    \caption{The outer binary tree mechanism}
    \label{fig:binary_tree}
\end{figure}

\subsection{Analyzing the binary tree mechanism}
Our main result can be stated as follows:
\begin{theorem}
The binary tree mechanism above implements the lexicographic maxmin solution.
\end{theorem}
\begin{proof}
We prove the above theorem in two parts: 
\begin{enumerate}[(A)]
\item the solution $\bfu^*$ is a subgame perfect outcome, and
\item $\bfu^*$ is the only possible subgame perfect outcome. 
\end{enumerate} 

\emph{Part (A):} Assume all players chooses $\bfu^*$, except player 1 who chooses some other outcome $\bfu$. As a consequence of Corollary \ref{thm:knockout_gen}, the outcome $\bfu^*$ is implemented in the knockout game $K(S_1,S_2,\mathcal{N})$ that is played at the parent node of $P_1$, where $S_1 = P_1 = \bfu$ and $S_2 = P_2 = \bfu^*$. Since the outcome at every other leaf node, besides $S_1$ and $S_2$, is $\bfu^*$, a repeated application of Lemma \ref{lem:KMsubgames} tells us that the outcome at $S_{1,2,\ldots,n}$ is $\bfu^*$. 

\emph{Part (B):} Assume we have a subgame perfect equilibrium where not all players choose $\bfu^*$ in the first round. In any such equilibrium, any player $i$ is guaranteed a utility of at least $u^*_i$ -- for if not, player $i$ can choose $\bfu^*$, and by a repeated application of Corollary \ref{thm:knockout_gen} and Lemma \ref{lem:KMsubgames}, we have that $\bfu^*$ is the outcome implemented in every subgame node that is reachable from $P_i$: ultimately at $S_{1,2,\ldots,n}$. Since every player $i$ is guaranteed at least $u^*_i$ in a subgame perfect equilibrium, the only possible outcome is $\bfu^*$ (by Lemma \ref{lem:unique}).  \qed
\end{proof}

\section{Discussion}
In this paper, we look at a widely studied and well characterized bargaining solution, namely the lexicographic maxmin solution. Although a lot of work has been done on the implementation of other major bargaining solutions, there is no known mechanism for the implementation of the lexicographic maxmin solution. Interestingly, we are able to design such a mechanism by uncovering and utilizing novel combinatorial properties of the lexicographic maxmin solution. In doing so, we also make use of the standard assumption in bargaining implementation theory that the space of outcomes is such that it is possible for each player to grab all the surplus. As we pointed out earlier, an interesting open problem is to implement any of the various bargaining solutions, including the lexicographic maxmin solution, without the above assumption, or to prove that such implementations are impossible -- our intuition aligns more with the latter (details deferred to a full version). It would also be interesting to look at the case where there is uncertainty in the knowledge that each player has about the utility valuations of other players. We hope this paper fuels more research in the abovementioned directions.

\newpage

\newpage

\appendix

\section{The Kalai-Smorodinsky Solution}\label{app:KS}
The Kalai-Smorodinsky (KS for short) solution is defined as follows:
\begin{definition}[Kalai-Smorodinsky solution]
Given a utility profile $\bfu$, the KS solution is given by the lottery $l \in P(A)$ that satisfies:
\begin{align*}
\forall ~j \in N, ~ \frac{U_j(l) - U_j(s^*)}{ \underset{x \in \bfI(\mathbf{u})}{\max} U_j(x) - U_j(s^*)} = \max_{l^\prime \in I(\mathbf{u})} \min_{i \in N} \frac{U_i(l^\prime) - U_i(s^*)}{ \underset{x \in \bfI(\mathbf{u})}{\max} U_i(x) - U_i(s^*)}.
\end{align*}
\end{definition}

In other words, the KS solution is the maximal point that results in the equalization of the relative utility gains among the agents. The KS solution is obviously individually rational, and it is easy to see that it is invariant under affine transformations.

The KS solution is not guaranteed to be Pareto optimal. We illustrate this point with an example.

\begin{example}\label{eg:roth2}
Consider a three person bargaining problem whose disagreement point $s^*$ is such that the utilities of the agents for it is $(0,0,0)$. The set of alternatives is given by $s^*$ and five other points $a,b,c,d,e$ for which the utilities are $(1,0,0)$, $(0,1,0)$, $(0,0,1)$, $(1,1,0)$, $(1,0,1)$ respectively. \footnote{This example violates Assumption \ref{ass:favorite2}. We use it nevertheless for ease of exposition.} Clearly the set of Pareto optimal lotteries are those that mix between $d$ and $e$, and any such must necessarily give utility $1$ to the third agent. It is easy to see that the KS solution here is one that gives utilities $(1/2,1/2,1/2)$ to the agents, which is Pareto dominated. 
\end{example}

In fact, it is known that for bargaining problems with $3$ or more agents, no solution exists that satisfies the axioms of Pareto efficiency and symmetry together with monotonicity \cite{roth1979impossibility}.

\section{Subgame perfect reduction: Proof of Lemma \ref{lem:KMsubgames}}\label{app:subgames}
\begin{proof}
Given that $\alpha \relD \beta$, we know that $\alpha$ is the unique subgame perfect outcome of $K(\alpha,\beta, \mathcal{N})$. 

Consider the extensive form game tree of the perfect information game $K(\alpha,\beta, \mathcal{N})$. Some of the terminal nodes, say the set $\mathcal{A}$, give the utility outcome $\alpha$, while some, say the set $\mathcal{B}$, others give the utility outcome $\beta$.
Given any two bargaining games $A$ and $B$ that each produce some outcome, the game tree of the Knockout mechanism $K(A,B,\mathcal{N})$ by appending the game trees of $A$ and $B$ to the nodes in $\mathcal{A}$ and $\mathcal{B}$, respectively.

Now consider a node $a \in \mathcal{A}$ in the game tree of $K(A,B,\mathcal{N})$. Since $\alpha$ is the unique subgame perfect outcome of game $A$, we know that a backward induction on the subtree below $a$, upto $a$, yields the outcome $\alpha$. Similarly, for any node $b \in \mathcal{B}$, the same reduction yields the outcome $\beta$.

Since the reduced tree obtained from the above step is identical to that of $K(\alpha,\beta, \mathcal{N})$, and $\alpha$ is the unique subgame perfect equilibrium of $K(\alpha,\beta, \mathcal{N})$, a backward induction on this reduced tree yields $\alpha$.

As a result, it must be that a backward induction on the full game tree of $K(A,B,\mathcal{N})$ also yields $\alpha$, whence the claim is proved. \qed
\end{proof}
%\section{Difference from cake-cutting}\label{app:cake}
%We want to exhibit a space of utility outcomes that is possible under our model but not as an instance of the cake-cutting problem.
%
%For the case of two players, let the space of outcomes be given by the convex hull of the following points:
%\begin{align*}
%(0,0), (1,0), (0,1) \mbox{ and } (0.75,0.75).
%\end{align*}
%Clearly, this is possible under our setting.
%
%Assume there is cake-cutting problem with the same outcome space. Obviously, $(0,1)$ and $(1,0)$ correspond to the outcome where the whole cake $[0,1]$ goes to one of the two players.
%Now consider the other two outcomes. We can divide the cake into four parts $a,b,c,d$ such that 
%\begin{itemize}
%\item in the outcome $(0,0)$ player 1 gets $a,b$ and player 2 gets $c,d$, and,
%\item in the outcome $(0.75,0.75)$ player 1 gets $a,c$ and player 2 gets $b,d$.
%\end{itemize}
%
%Let $u_i(t)$ denote the utility of player $i$ for part $j$ of the cake. Then we have:
%\begin{align*}
%u_1(a) + u_1(b) = 0, \\
%u_1(a) + u_1(c) = 0.75, \\
%u_2(c) + u_2(d) = 0, \\
%u_2(b) + u_2(d) = 0.75.
%\end{align*}
%
%$u_1(a) + u_1(b) = 0$ implies $u_1(a) = 0$ and $u_1(b) = 0$. Similarly, $u_2(c) = u_2(d) = 0$. 
%Therefore $u_1(c) = 0.75$ and $u_2(b) = 0.75$. Since the whole cake is worth 1 for each, we immediately get $u_1(d) = u_2(a) = 0.25$.
%
%Using the above observations see that if we give $c,d$ to player 1 and $a,b$ to player 2, then both get a utility of 1. But we know that $(1,1)$ is infeasible, and hence reach a contradiction.

\end{document}